\documentclass[11pt,letterpaper]{article}
\usepackage{srcltx}

\usepackage{cite,fullpage,amsthm}

\usepackage{amssymb,amsmath,eucal}
\usepackage{amsfonts}
\usepackage{indentfirst}
\usepackage{graphicx}
\newcommand{\beq}{\begin{equation}}
\newcommand{\eeq}{\end{equation}}

\newtheorem{theorem}{Theorem}

\newcommand{\barr}{\begin{array}}
\newcommand{\earr}{\end{array}}

\begin{document}
\title{General solution of a second order non-homogenous linear
difference equation with noncommutative coefficients }
\author{\normalsize M A
Jivulescu$^{1,2}$, A Napoli$^1$, A Messina$^1$\\
\small ${}^1$ MIUR, CNISM and Dipartimento di Scienze Fisiche ed
Astronomiche, \\\small Universit\`{a} di Palermo, via Archirafi
36, 90123 Palermo, Italy\\\small ${}^2$  Department of
Mathematics, " Politehnica" University of Timi\c{s}oara,
\\\small P-ta Victoriei Nr. 2,
 300006 Timi\c{s}oara,
Romania\\
\small {\it Email Address:$^{1,2}$ maria.jivulescu@mat.upt.ro}\\
\small {\it Email Address:$^1$ messina@fisica.unipa.it}}

\thispagestyle{empty}

\maketitle

\begin{center}SUMMARY\end{center}
\noindent The detailed construction of the general solution of a
second order non-homogenous linear operator-difference equation is
presented. The wide applicability of such an equation as well as
the usefulness of its  resolutive formula is shown by studying
some applications belonging to different mathematical contexts.

\medskip


%%%%%            Providing 2000 MSC codes is optional for AMIS       %%%%%%%%%%%%%

%\medskip
%\noindent {\bf 2000 Mathematics Subject Classification}: 94A17, 94A15.

\section{INTRODUCTION}
In this paper we report the explicit representation of the general
solution of the second order non-homogenous linear
operator-difference equation
\begin{equation}\label{NHE} Y_{p+2}=\mathcal{L}_0Y_p+\mathcal{L}_1Y_{p+1}+\phi_{p+1},\end{equation}
 where the unknown  $\{Y_p\}_{p\in\mathbb{N}}$ as well as the non-homogenous term $\{\phi_p\}_{p\in \mathbb{N}}$ are sequences from a vectorial space $V$,
and the coefficients $\mathcal{L}_0,\mathcal{L}_1$, are linear
noncommutative operators mapping $V$ on itself, independent from
the discrete  variable $p\in \mathbb{N}$. This equation
encompasses interesting problems arising in very different
scenarios. If, for instance,  the reference space $V$ is the
complex Euclidean space $\mathbb{C}^n$, that is $Y_p$ and
$\phi_{p+1}$ are $n$-dimensional vectors, $\mathcal{L}_0$ and
$\mathcal{L}_1$ $n\times n$ complex matrices, then eq. (\ref{NHE})
is the vectorial representation of a system of second-order linear
non-homogenous difference equations. As another example, let's
identify $V$ as the vectorial space of all linear operators
defined on a given Hilbert
    space.  Now, the operators $\mathcal{L}_0$ and $\mathcal{L}_1$ act upon
operators and
 for this reason are called superoperators.
  The master equations  appearing in the theory of open quantum systems  provide
  examples
 of equations belonging to this class\cite{Breuer}. It is of relevance to
 emphasize from the very beginning that the ingredients  $Y_p$, $\phi_p$,
 $\mathcal{L}_0$ and $\mathcal{L}_1$ of eq. \eqref{NHE} may be also interpreted as
 elements of an
 assigned algebra $V$.
 Let's consider, for
example, $V$ as the noncommutative algebra of all square matrices
of order $n$, that is $M_n[\mathbb{C}]$. Then,  eq. (\ref{NHE})
defines a second order non-homogenous linear matrix-difference
equation, where $Y_p, \phi_p, \mathcal{L}_0$ and $\mathcal{L}_1$
belong to $M_n(\mathbb{C})$. We wish further emphasis that if $V$
is the vectorial space of the smooth functions over an interval
$I$, that is $C^\infty(I)$, then eq. (\ref{NHE}) represents a wide
class of functional-difference
    equations\cite{Kolmanovskii}, including
    difference-differential equations or
    integro-difference equations\cite{Pinney, Bellman, Driver}.
These few examples
 motivate the interest toward  the search  of techniques for
 solving the operator eq.(\ref{NHE}), with $\mathcal{L}_0,\mathcal{L}_1$  noncommutative
 coefficients. \par In this paper  we cope with  such a problem and succeed in giving its
explicit solution  leaving unspecified the abstract "support
space" wherein eq. (\ref{NHE}) is formulated. This means that  we
do not choose from the very beginning the mathematical nature
 of its ingredients, rather we only require that all the symbols
 and operations appearing in eq. (\ref{NHE}) are meaningful.
 Accordingly,
 \lq\lq
 vectors\rq\rq \hspace{0.1cm} $Y_p$ may be added, this operation being commutative
 and,
 at the same time, may be acted upon by $\mathcal{L}_0$ or
 $\mathcal{L}_1$ ( hereafter called operators) transforming themselves into other \lq\lq vectors\rq\rq of $V$. The
 symbol $Y_0=0$   simply denotes, as usual, the neutral element of
 the underlying space. Finally we put $(\mathcal{L}_a\mathcal{L}_b)Y\equiv \mathcal{L}_a(\mathcal{L}_bY)\equiv
 \mathcal{L}_a\mathcal{L}_bY$ with $a$ or $b=0,1$ and define
 addition between operators through linearity.
\par The paper is organized as follows.  \\ The first section
presents the solution of an arbitrary Cauchy problem associated
with eq. \eqref{NHE}. Some interesting consequences of such a
result are derived in the subsequent section. The practical
usefulness of our resolutive formula is shown in the third section
where we solve some nontrivial functional-difference and
integral-difference equations. Some concluding remarks are
presented in the last section.
\section{EXPLICIT CONSTRUCTION OF THE RESOLUTIVE FORMULA OF EQ.\eqref{NHE}}
Let's begin by recalling that if
 $\{Y_p^{*}\}_{p\in\mathbb{N}}$ and   $\{Y_p\}_{p\in\mathbb{N}}$
are solutions  of
 eq.(\ref{NHE}), then $\{Y_p^{H}\}_{p\in\mathbb{N}}$ defined as $Y_p-Y_p^*\equiv Y_p^{H}$  is a solution of the associated homogenous equation
 \begin{equation}\label{HEQ} Y_{p+2}=\mathcal{L}_0Y_p+\mathcal{L}_1Y_{p+1}\end{equation}
Thus, as for the linear differential equations, and independently
from the noncommutative nature of $\mathcal{L}_0$ and
$\mathcal{L}_1$, solving eq. \eqref{NHE} amounts at being able to
construct the general integral of eq. \eqref{HEQ} and to find out
a particular solution of eq. \eqref{NHE}. To this end, we start
with the following theorem which extends
 a recently published result\cite{Jivulescu} concerning the exact
resolution of the following Cauchy
problem\begin{equation}\label{CP}\left\{\begin{array}{rl}
      Y_{p+2}=\mathcal{L}_0Y_p+\mathcal{L}_1Y_{p+1}, \\
       Y_0=0,\quad Y_1=B\end{array}\right..\end{equation}

\begin{theorem}
The  solution of the Cauchy problem
\begin{equation}\label{HE}\left\{\begin{array}{rl}
      Y_{p+2}=\mathcal{L}_0Y_p+\mathcal{L}_1Y_{p+1}\\
       Y_0=A,\quad Y_1=B\end{array}\right., \end{equation}
 can be
written as
\begin{equation}\label{solH}Y_p^{(H)}=\alpha_pA+\beta_pB,\end{equation}
 where  the operators  $\alpha_p$ and $\beta_p$ have the following
 form
\begin{equation}\label{alfa} \alpha_p=\left\{\begin{array}{ll}
\sum\limits_{t=0}^{[\frac{p-2}{2}]}\{\mathcal{L}_0^{(t)}\mathcal{L}_1^{(p-2-2t)}\}\mathcal{L}_0&\quad if \quad p\geq2 \\
 0&\quad if \quad  p=1\\
E&\quad if \quad p=0
\end{array}\right.,\end{equation}
\begin{equation}\label{beta}
\beta_p=\left\{\begin{array}{ll}\sum\limits_{t=0}^{[\frac{p-1}{2}]}\{\mathcal{L}_0^{(t)}\mathcal{L}_1^{(p-1-2t)}\}
& \quad if  \quad p\geq 2\\E&\quad if \quad  p=1\\ 0& \quad if
\quad p=0
\end{array}\right.,
\end{equation}
\end{theorem}

We recall that the mathematical symbol $\{\mathcal{L}_0
^{(u)}\mathcal{L}_1^{(v)}\}$, in accordance with ref
\cite{Jivulescu}, denotes the sum of all possible distinct
permutations of $u$ factors $\mathcal{L}_0$ and $v$ factors
 $\mathcal{L}_1$, while $0$, $E:V\rightarrow V$ define the null
 and
  the identity operator in $V$, respectively.
We omit the proof of this theorem since it is practically
coincident with that given in ref \cite{Jivulescu}. Here instead
we demonstrate the following
\begin{theorem} Eq. (\ref{NHE}) admits the particular solution

\begin{equation}\label{SOLP}
Y_p^*=\left\{\begin{array}{ll}
\sum\limits_{r=1}^{p-1}\beta_{p-r}\phi_r, &\quad if \quad p\geq2 \\
 0,&\quad if \quad  p=0,1\\

\end{array}\right.
\end{equation}
\end{theorem}
\begin{proof}
It is immediate to verify, by direct substitution, that  the
sequence given by eq. (\ref{SOLP}) satisfies eq. (\ref{NHE})
written for $p=0$ and $p=1$. To this end, it is enough to exploit
eqs. \eqref{beta} and \eqref{SOLP} getting
$Y_2^*=\beta_1\phi_1=\phi_1$ and
$Y_3^*=\beta_2\phi_1+\beta_1\phi_2=\mathcal{L}_1\phi_1+\phi_2$.
\\For a generic $p\geq 2$,  introducing $Y_p^*$ in the right
hand of eq. (\ref{NHE}) yields

\begin{eqnarray}\label{10}\mathcal{L}_0\sum\limits_{r=1}^{p-1}\beta_{p-r}\phi_r+\mathcal{L}_1\sum\limits_{r=1}^{p}\beta_{p+1-r}\phi_r+\phi_{p+1}\nonumber\\
=\label{OPT}\sum\limits_{r=1}^{p-1}(\mathcal{L}_0\beta_{p-r}+\mathcal{L}_1\beta_{p+1-r})\phi_r+\mathcal{L}_1\beta_1\phi_{p}+\phi_{p+1}
\end{eqnarray}
Applying theorem (1) to the Cauchy problem expressed by eq.
\eqref{CP}, we easily deduce that for $p\geq 2$  and $r=1,2,\dots,
p-1$ the following operator identity
\begin{equation}
\mathcal{L}_0\beta_{p-r}+\mathcal{L}_1\beta_{p+1-r}=\beta_{p+2-r},\quad
\end{equation}holds. Thus, the  expression
given by eq. \eqref{10} may be cast as follows

\begin{eqnarray}\label{11}
\sum\limits_{r=1}^{p-1}\beta_{p+2-r}\phi_r+\beta_{p+2-(p)}\phi_{p}+\beta_{p+2-(p+1)}\phi_{p+1}=\sum\limits_{r=1}^{p+1}\beta_{p+2-r}\phi_r
\end{eqnarray}
where we have exploited the identity
$\beta_2=\mathcal{L}_1\beta_1$ based on  eq. (\ref{beta}). Since
the right hand of eq. \eqref{11} coincides with $Y_{p+2}^*$ as
given by eq. (\ref{SOLP}), we may conclude that $\{Y_p^*\}_{p\in
\mathbb{N}}$, expressed by eq. \eqref{SOLP}, provides a particular
solution of eq. \eqref{NHE}.
\end{proof}
On the basis of theorem (1) and (2) we hence may state our main
result, that is
\begin{theorem}The solution of the Cauchy problem
\begin{equation}\label{CPPNH}\left\{\begin{array}{rl}
      Y_{p+2}=\mathcal{L}_0Y_p+\mathcal{L}_1Y_{p+1}+\phi_{p+1}\\
       Y_0=A,\quad Y_1=B\end{array}\right., \end{equation} is
\begin{equation}\label{solNH}Y_p=\alpha_pA+\beta_pB+\sum\limits_{r=1}^{p-1}\beta_{p-r}\phi_r\end{equation}
where $A$ and $B$ are generic admissible initial conditions and
$\alpha_p$ and $\beta_p$ are defined by eqs. \eqref{alfa} and
\eqref{beta}, respectively.
 \end{theorem} We emphasize that eq. \eqref{solNH} furnishes a recipe to solve
 explicitly,
 that is in terms of its ingredients $\mathcal{L}_0,
 \mathcal{L}_1$ and $\{\phi_{p+1}\}_{p\in\mathbb{N}}$, the general
 Cauchy
 problem expressed by
 eq. \eqref{CPPNH}.  In the subsequent sections we will highlight that our result is effectively exploitable, providing indeed
 a useful approach to solve
problems belonging to very different mathematical contexts. This
circumstance adds a further robust motivation to investigate eq.
\eqref{NHE} and its consequences.

We conclude this section looking for the structural form assumed
by eq. \eqref{solNH} solely relaxing  the noncommutativity between
the two operator coefficients $\mathcal{L}_0$ and $\mathcal{L}_1$.
To this end, it is useful to recall the definition of the
 Chebyshev polynomials of the second kind
$\mathcal{U}_p(x), x\in \mathbb{C}$\cite{SCH}
\begin{equation}\label{Cby}\mathcal{U}_p(x)=\sum\limits_{m=0}^{[p/2]}(-1)^m\frac{(p-m)!}{m!(p-2m)!}(2x)^{p-2m}\end{equation} Indeed, taking into consideration
 that the number of all the different terms appearing in the
operator symbol $\{\mathcal{L}_0 ^{(u)}\mathcal{L}_1^{(v)}\}$
coincides with the binomial coefficient $ \left(\begin{array}{c}
                                u+v\\ m
                                \end{array} \right)$, with $m=min(u,v)$ as well as  assuming the existence
 of the operator $(-\mathcal{L}_0)^{-\frac{1}{2}}$, then the operators $\alpha_p$ and $\beta_p$ for $p\geq
 2$ may be cast as follows
\begin{equation}\label{C11}\alpha_p=-(-\mathcal{L}_0)^{\frac{p}{2}}\mathcal{U}_{p-2}\left(\frac{1}{2}\mathcal{L}_1(-\mathcal{L}_0)^{-\frac{1}{2}}\right)\end{equation}
and
\begin{equation}\label{C22}\beta_p=(-\mathcal{L}_0)^{\frac{p-1}{2}}\mathcal{U}_{p-1}\left(\frac{1}{2}\mathcal{L}_1(-\mathcal{L}_0)^{-\frac{1}{2}}\right)\end{equation}
 where
$\mathcal{U}_p\left(\frac{1}{2}\mathcal{L}_1(-\mathcal{L}_0)^{-\frac{1}{2}}\right)$
means the operator value of the polynomial $\mathcal{U}_p$ defined
in accordance with eq.\eqref{Cby} for
$x=\left(\frac{1}{2}\mathcal{L}_1(-\mathcal{L}_0)^{-\frac{1}{2}}\right)$.
Thus, the solution of Cauchy problem \eqref{CPPNH} may be
rewritten, for $p\geq 2$, as
\begin{eqnarray}\label{POWC}Y_p=(-\mathcal{L}_0)^{\frac{p-1}{2}}\mathcal{U}_{p-1}\left(\frac{1}{2}\mathcal{L}_1(-\mathcal{L}_0)^{-\frac{1}{2}}\right)B-(-\mathcal{L}_0)^{\frac{p}{2}}\mathcal{U}_{p-2}\left(\frac{1}{2}\mathcal{L}_1(-\mathcal{L}_0)^{-\frac{1}{2}}\right)A+\nonumber\\+\sum\limits_{r=1}^{p-1}(-\mathcal{L}_0)^{\frac{p-r-1}{2}}\mathcal{U}_{p-r-1}\left(\frac{1}{2}\mathcal{L}_1(-\mathcal{L}_0)^{-\frac{1}{2}}\right)\phi_{r}\end{eqnarray}
where $Y_0=A$ and $Y_1=B$ are the prescribed initial conditions.

\section{SOME CONSEQUENCES OF THE RESOLUTIVE FORMULA}
The mathematical literature offers several ways of solving linear
second  difference equations such as the matrix method or the
generating function method. In the following we will heuristically
generalized these methods to the operator case. The novelty of our
method enables to deduce, by comparison with these approaches,
some interesting consequent identities. Indeed,  the second-order
operator difference equation \eqref{NHE} may be traced back to the
first-order vectorial representation
\begin{equation}\label{ME}\mathbf{Y}_{p+1}=C_1\mathbf{Y}_p+\Phi_{p+1}\end{equation}
where $\mathbf{Y}_p=\left(\begin{array}{c}
                                Y_{p}\\Y_{p+1}
                                \end{array} \right)$, $C_1=\left(\begin{array}{cc}0&E\\
                                \mathcal{L}_0&\mathcal{L}_1
                                \end{array} \right)$, $\Phi_{p+1}=\left(\begin{array}{c}
                                0\\\phi_{p+1}
                                \end{array} \right)$,  $\mathbf{Y}_0=\left(\begin{array}{c}
                                A\\B
                                \end{array} \right)$.

 Successive iterations easily  lead us to the formal solution
\begin{equation}\label{c2}\mathbf{Y}_p=C_1^{p}\mathbf{Y}_0+\sum\limits_{r=1}^{p}C_1^{p-r}\Phi_{r}\end{equation}
On this basis, the solution of eq.\eqref{NHE} may be written as
\cite{Gohberg}
\begin{equation}Y_p=P_1C_1^{p}\mathbf{Y}_0+P_1\sum\limits_{r=1}^{p}C_1^{p-r}R_1\Phi_{r}\end{equation}
where $P_1=(E \quad 0)$ and $R_1=\left(\begin{array}{c}
                                0\\E
                                \end{array} \right)$.
This solution is of practical use only if we are able to evaluate
the general integer power of the companion matrix $C_1$.
Exploiting our procedure of writing the solution of eq.
\eqref{NHE}, the vector $\mathbf{Y}_p$ may be expressed,
accordingly with eq. (\ref{solNH}), in terms of operator sequences
$\alpha_p$ and $\beta_p$ like
\begin{equation}\label{c1}\mathbf{Y}_p=\left(\begin{array}{c}
                                \alpha_{p}A+\beta_{p}B+\sum\limits_{r=1}^{p-1}\beta_{p-r}\phi_r\\\alpha_{p+1}A+\beta_{p+1}B+\sum\limits_{r=1}^{p}\beta_{p+1-r}\phi_r
                                \end{array} \right)=\left(\begin{array}{cc}
                               \alpha_{p}&\beta_{p}\\\alpha_{p+1}&\beta_{p+1}
                                \end{array} \right)\left(\begin{array}{c}
                                A\\B
                                \end{array} \right)+\left(\begin{array}{c}
                                \sum\limits_{r=1}^{p-1}\beta_{p-r}\phi_r\\\sum\limits_{r=1}^{p}\beta_{p+1-r}\phi_r
                                \end{array} \right)\end{equation}
Confining  ourselves to the homogenous version of eq. \eqref{NHE},
that is putting $\phi_{p+1}=0$  into eqs. \eqref{c2} and
\eqref{c1}, we get the formula for the $p$-th power of the
companion matrix $C_1$ as follows

\begin{equation}\label{PowC1}C_1^p=\left(\begin{array}{cc}
                               0&E\\ \mathcal{L}_0&\mathcal{L}_1
                                \end{array} \right)^p=\left(\begin{array}{cc}
                              \alpha_{p}&\beta_{p}  \\\alpha_{p+1}&\beta_{p+1}
                                \end{array} \right)\end{equation}

Another possible way of treating  eq.\eqref{NHE} is via the
generating functions method\cite{Goldberg, Hildebrand}. We recall
that, given the sequence $\{Y_p\}_{p\in \mathbb{N}}$,
 the associated generating function  $Y(s), s\in \mathbb{C}$ is defined as
\begin{equation}\label{Gen}\mathcal{G}Y_p\equiv Y(s)\equiv\sum\limits_{p=0}^{\infty}Y_ps^p\end{equation}
under the assumption that the series converges when $|s|\leq \xi$,
for some positive number $\xi$. The advantage of this method
consists in the  systematical possibility of transforming a
difference equation in an algebraic one in the unknown $Y(s)$. In
order to apply such approach to the operator-difference equation
given by \eqref{NHE}, we stipulate that
$\mathcal{G}[\mathcal{L}_iY_p]$, $\mathcal{L}_i(\mathcal{G}Y_p)$
are both defined and that
$\mathcal{G}[\mathcal{L}_iY_p]=\mathcal{L}_i(\mathcal{G}Y_p),
i=0,1$. Accordingly, heuristically, we transform both sides of eq.
\eqref{NHE} getting
\begin{eqnarray}\nonumber \frac{Y(s)-A-Bs}{s^2}=\mathcal{L}_1\frac{Y(s)-A}{s}+\mathcal{L}_0Y(s)+\Phi(s)\end{eqnarray}
Thus, assuming the existence of
$(E-\mathcal{L}_1s-\mathcal{L}_0s^2)^{-1}$ within the convergence
disk, we have
\begin{equation}Y(s)=(E-\mathcal{L}_1s-\mathcal{L}_0s^2)^{-1}A+(E-\mathcal{L}_1s-\mathcal{L}_0s^2)^{-1}(B-\mathcal{L}_1A)s+(E-\mathcal{L}_1s-\mathcal{L}_0s^2)^{-1}\Phi(s)s^2\end{equation}
or equivalently
\begin{eqnarray}Y(s)=(E-\mathcal{L}_1s-\mathcal{L}_0s^2)^{-1}[E-\mathcal{L}_1 s]A+(E-\mathcal{L}_1s-\mathcal{L}_0s^2)^{-1}Bs+(E-\mathcal{L}_1s-\mathcal{L}_0s^2)^{-1}\Phi(s) s^2\end{eqnarray}
On the other hand, accordingly with eqs. \eqref{Gen} and
\eqref{solNH} it holds that
\begin{equation}Y(s)=\sum\limits_{p=0}^{\infty}(\alpha_pA+\beta_pB+\sum\limits_{r=1}^{p-1}\beta_{p-r}\phi_r)s^p\end{equation}
Thus, one notes that imposing  $\phi_{p+1}\equiv 0$ and $B=0$,
respectively $A=0$, we heuristically find
 the generating function of the operator
sequences $\alpha_p$  and $\beta_p$ in the closed form as
\begin{equation}\label{GenA}\mathcal{G}\alpha_p\equiv\sum\limits_{p=0}^{\infty}\alpha_ps^p:=(E-\mathcal{L}_1s-\mathcal{L}_0s^2)^{-1}[E-\mathcal{L}_1s]\end{equation}
respectively
\begin{equation}\label{GenB}\mathcal{G}\beta_p\equiv\sum\limits_{p=0}^{\infty}\beta_ps^p:=(E-\mathcal{L}_1s-\mathcal{L}_0s^2)^{-1}s\end{equation}
The particular case  $\mathcal{L}_0=-E$ reproduces  the generating
functions of the Chebyshev polynomials of second kind. Extracting
indeed for the sake of convenience the first two terms of the
series, that is writing $\nonumber
\sum\limits_{p=0}^{\infty}\alpha_ps^p=\alpha_0+\alpha_1s+\sum\limits_{p=2}^{\infty}\alpha_ps^p$
with the help of eq. \eqref{C11} and $\alpha_0=E, \alpha_1=0$ we
 get
\begin{eqnarray}\label{last}(E-\mathcal{L}_1s+s^2)^{-1}[E-\mathcal{L}_1s]=E-\sum\limits_{p=2}^{\infty}\mathcal{U}_{p-2}[\frac{\mathcal{L}_1}{2}]s^p
\end{eqnarray}
The eq. \eqref{last} easily determines the generating function of
the sequence $\{U_p[\mathcal{L}_1/2]\}_{p\in \mathbb{N}}$ in the
form
\begin{eqnarray}\nonumber\sum\limits_{p=2}^{\infty}\mathcal{U}_{p-2}[\frac{\mathcal{L}_1}{2}]s^p=(E-\mathcal{L}_1s+s^2)^{-1}[E-\mathcal{L}_1s+s^2-E+\mathcal{L}_1s]\\
\Leftrightarrow
\sum\limits_{p=0}^{\infty}\mathcal{U}_{p}[\frac{\mathcal{L}_1}{2}]s^p=(E-\mathcal{L}_1s+s^2)^{-1}E\end{eqnarray}
The novel results obtained in this paper exploiting our resolutive
formula ( eqs. \eqref{PowC1}, \eqref{GenA}, \eqref{GenB}) clearly
evidence that our recipe to manage eq.\eqref{NHE} successfully
integrate with other resolutive methods. Thus, we may claim that
our resolutive formula do not possess a formal character only,
since it helps to provide new interesting identities.
\section{APPLICATIONS OF OUR RESOLUTIVE FORMULA}
\vspace{0.5cm} \textit{An example of matrix-difference equation coped with our formula}\vspace{0.5cm}\\
Let us consider the second order matrix-difference eqution
\begin{eqnarray} Y_{p+2}=M_0Y_p+M_1Y_{p+1}+\Phi_{p+1}\end{eqnarray}
where $M_0,M_1\in \mathcal{M}_n(\mathbb{C})$ are noncommutative
nilpotent matrices of index $2$, that is $M_i^2=0, \quad i=0,1$.
Prescribing the initial conditions $Y_0=A, Y_1=B$, then the
solution of this equation is given by the eq. (\ref{solNH}).  The
analysis of the matrix term $\{M_0^{(u)}M_1^{(v)}\}$ which appears
in the composition of the matrix-operator $\alpha_p$ and $\beta_p$
brings to light interesting peculiarities due to the specific
nature of the coefficients $M_0$ and $M_1$. By definition, the
term $\{M_0^{(u)}M_1^{(v)}\}$ represents the sum of all possible
terms of $u$ factor $M_0$ and $v$ factors $M_1$. Thus, it is quite
simple to deduce that now the matrix-term of the form
$M_0^{\nu_1}M_1^{\nu_2}M_0^{\nu_3}M_1^{\nu_4}\dots$ is equal with
zero, if $\nu_i>1, (\forall) i$. Hence, we deduce that the
operator $\{M_0^{(u)}M_1^{(v)}\}$ survives solely when $u=v$ or
$v=u\pm 1$. Indeed, when $u=v, u\geq2$, then in the sum
$\{M_0^{(u)}M_1^{(u)}\}$  survive only the terms
$\underbrace{[M_0M_1][M_0M_1]\dots [M_0M_1]}_{u-times}$ and
$\underbrace{[M_1M_0][M_1M_0]\dots [M_1M_0]}_{u-times}$. Further,
the only nonvanishing   matrix-terms $\{M_0^{(u)}M_1^{(u+1)}\}$
are $M_1\underbrace{[M_0M_1][M_0M_1]\dots[M_0M_1]}_{u-times}$, as
well as from $\{M_0^{(u)}M_1^{(u-1)}\}$ the terms
$\underbrace{[M_0M_1][M_0M_1]\dots[M_0M_1]}_{(u-1)-times}M_0$,
respectively. Exploiting the above results for the matrix term
$\{M_0^{(t)}M_1^{(p-1-2t)}\}$ we establish that\vspace{0.5cm}
\begin{eqnarray}\beta_p=\left\{\begin{array}{rl}
      \{M_0^{(\frac{p-1}{3})}M_1^{(\frac{p-1}{3})}\}=\underbrace{[M_0M_1]\dots[M_0M_1]}_{k-times}+\underbrace{[M_1M_0]\dots[M_1M_0]}_{k-times}, \quad p=3k+1
      \\\vspace{0.5cm}
      \{M_0^{(\frac{p-2}{3})}M_1^{(\frac{p+1}{3})}\}=M_1\underbrace{[M_0M_1]\dots
      [M_0M_1]}_{k-times},\quad p=3k+2
      \\\vspace{0.5cm}
      \{M_0^{(\frac{p}{3})}M_1^{(\frac{p-3}{3})}\}=\underbrace{[M_0M_1]\dots
      [M_0M_1]}_{(k-1)-times}M_0,\quad p=3k,
       \end{array}\right.\end{eqnarray}
where $ k=1,2,\dots$. Similarly, we have that
\begin{eqnarray}\alpha_p=\left\{\begin{array}{rl}\vspace{0.2cm}
     \{M_0^{(\frac{p-2}{3})}M_1^{(\frac{p-2}{3})}\}M_0=\underbrace{[M_0M_1]\dots[M_0M_1]}_{k-times}M_0,\quad p=3k+2 \\\vspace{0.2cm}
       \{M_0^{(\frac{p-3}{3})}M_1^{(\frac{p}{3})}\}M_0=M_1\underbrace{[M_0M_1]\dots[M_0M_1]}_{(k)-times}M_0,\quad
       p=3k+3\\\vspace{0.2cm}
       \{M_0^{(\frac{p-1}{3})}M_1^{(\frac{p-4}{3})}\}M_0=(\underbrace{[M_0M_1]\dots [M_0M_1]}_{(k-1)-times}M_0)M_0=0,\quad p=3k+1\end{array}\right.\end{eqnarray}
Hence, we may write the solution into a closed form
\begin{eqnarray}\nonumber
Y_p=\left[\delta_{\frac{p-2}{3},\left[\frac{p-2}{3}\right]}(M_0M_1)^{[\frac{p-2}{3}]}M_0+\delta_{\frac{p-3}{3},\left[\frac{p-3}{3}\right]}M_1(M_0M_1)^{[\frac{p-3}{3}]}M_0\right]A+\end{eqnarray}\begin{eqnarray}
[\delta_{\frac{p-1}{3},\left[\frac{p-1}{3}\right]}(M_0M_1)^{[\frac{p-1}{3}]}+\delta_{\frac{p-1}{3},\left[\frac{p-1}{3}\right]}(M_1M_0)^
{[\frac{p-1}{3}]}+\delta_{\frac{p-2}{3},\left[\frac{p-2}{3}\right]}M_1(M_0M_1)^{[\frac{p-2}{3}]}+\delta_{\frac{p-3}{3},\left[\frac{p-3}{3}\right]}(M_0M_1)
^{[\frac{p-3}{3}]}M_0]B\nonumber\end{eqnarray}
\begin{eqnarray}
+\sum\limits_{r=1}^{p-1}[\delta_{\frac{p-r-1}{3},\left[\frac{p-r-1}{3}\right]}\left((M_0M_1)^{[\frac{p-r-1}{3}]}+(M_1M_0)^{[\frac{p-r-1}{3}]}\right)+
\nonumber
\\\delta_{\frac{p-r-2}{3},\left[\frac{p-r-2}{3}\right]}M_1(M_0M_1)^{[\frac{p-r-2}{3}]}+\delta_{\frac{p-r-3}{3},\left[\frac{p-r-3}{3}\right]}(M_0M_1)^
{[\frac{p-r-3}{3}]}M_0]\Phi_r
\end{eqnarray}

\vspace{0.5cm} \textit{An example of functional-difference
equation coped with our method}\vspace{0.5cm} \\ The three-term
recurrence relation
\begin{equation}\label{ee}f_{p+2}(t)=-f_p(t-\tau_0)+f_{p+1}(t+\tau_1)\end{equation}
with the initial conditions $A=f_0(t)$ and $B=f_1(t)$ is an
example of   functional difference equation, traceable back to
 eq. \eqref{NHE}. It is indeed well-known that if $f(t)$ is a
 function of class $C^{\infty}$, then the translation of its
 independent variable from $t$ to $t+\tau$ can be represented as
 the effect on the same function of the operator $exp[\tau \frac{d\cdot}{dt}]=\sum\limits_{k=0}^{\infty}\frac{1}{k!}[\tau
 \frac{d\cdot}{dt}]^k$. This operator appears in a natural way when one
 studies problems characterized by translational invariance in a
 physical context\cite{Sakurai}.
Thus, by putting
$\mathcal{L}_{i}=(-1)^{i+1}exp[(-1)^{i+1}\tau_i\frac{d.}{dt}],i=0,1$
the commutativity property of the two operator coefficients
$\mathcal{L}_0$ and $\mathcal{L}_1$ allows us to write down the
solution of eq. \eqref{ee}  as follows
\begin{eqnarray}\nonumber
f_p(t)=\exp[-(\frac{p-1}{2})\tau_0\frac{d}{dt}]\mathcal{U}_{p-1}\left[\frac{1}{2}\exp[(\tau_1+\frac{\tau_0}{2})\frac{d}{dt}]\right]f_1(t)-\\
\exp[-(\frac{p}{2})\tau_0\frac{d}{dt}]\mathcal{U}_{p-2}\left[\frac{1}{2}\exp[(\tau_1+\frac{\tau_0}{2})\frac{d}{dt}]\right]f_0(t)\end{eqnarray}
Exploiting eq. \eqref{Cby} we may write down that
\begin{eqnarray}f_p(t)=\sum\limits_{k=0}^{[p-1/2]}(-1)^k \left(\begin{array}{c}
                                p-1-k\\k
                                \end{array} \right)f_1\left(t+(p-1-2k)\tau_1-k\tau_0\right)-\\
\sum\limits_{k=0}^{[p-2/2]}(-1)^k \left(\begin{array}{c}
                                p-2-k\\k
                                \end{array} \right)f_0\left(t+(p-2-2k)\tau_1-(k+1)\tau_0\right)\end{eqnarray}
Imposing, for example, the following initial conditions
$f_0(t)=e^{-t}$ and $f_1(t)=e^{t}$ we get
\begin{eqnarray}f_p(t)=
\exp[-(\frac{p-1}{2})\tau_0]\mathcal{U}_{p-1}\left[\frac{1}{2}\exp[(\tau_1+\frac{\tau_0}{2})]\right]e^t-\\
\exp[(
\frac{p}{2})\tau_0]\mathcal{U}_{p-2}\left[\frac{1}{2}\exp[-(\tau_1+\frac{\tau_0}{2})]\right]e^{-t}\end{eqnarray}
\vspace{0.5cm}\\
\textit{An example of integro-difference equation coped with our
method}\vspace{0.5cm}
\\Consider the difference-differential equation
\begin{eqnarray}f_{p+2}'(t)=\beta f_{p+1}(t)+\alpha f_{p}(t),\quad \alpha, \beta\in\mathbb{R}, \quad p=0,1,\dots\end{eqnarray}
where $f_p(t)$ is a $C^{\infty}(I)$ function with $f_0(t)$,
$f_1(t)$ and $\{f_p(0),\quad p=0,1,\dots\}$  prescribed
functions.\\ The above equation may be rewritten in the equivalent
form
\begin{eqnarray}\label{41}f_{p+2}(t)=\mathcal{L}_1f_{p+1}(t)+\mathcal{L}_0f_p(t)+f_{p+2}(0)\end{eqnarray}
where $\mathcal{L}_0=\alpha \mathcal{L}$, $\mathcal{L}_1=\beta
\mathcal{L}$ and $\mathcal{L}(\cdot)=\int\limits_0^t\cdot d\tau$.
Eq. \eqref{41} is a particular case of eq.\eqref{NHE}.
\\ The
explicit solution of this equation requires the knowledge of the
operator terms $\alpha_p$ and $\beta_p$. One remarks that
$\beta_p$ is the sum of $\left[\frac{p-1}{2}\right]+1$ operator
terms of the form $\{\mathcal{L}_0^t\mathcal{L}_1^{p-1-2t}\}$.
Because $\mathcal{L}_0=\alpha \mathcal{L}$ and
$\mathcal{L}_1=\beta \mathcal{L}$ then, for a finite $p$, holds
\begin{eqnarray}\{\mathcal{L}_0^k\mathcal{L}_1^{p-1-2k}\}=\alpha^k\beta^{p-1-2k} \left(\begin{array}{c}
                                p-1-k\\ k
                                \end{array} \right)\mathcal{L}^{p-1-k}\end{eqnarray}
Therefore, by direct substitution into eq. \eqref{beta} it follows
that
\begin{eqnarray}\beta_p=\sum\limits_{k=0}^{\left[\frac{p-1}{2}\right]}\alpha^k\beta^{p-1-2k} \left(\begin{array}{c}
                                p-1-k\\ k
                                \end{array} \right)\mathcal{L}^{p-1-k}\end{eqnarray}
Similarly, we have that
\begin{eqnarray}\alpha_p=\sum\limits_{k=0}^{\left[\frac{p-2}{2}\right]}\alpha^{k+1}\beta^{p-2-2k} \left(\begin{array}{c}
                                p-2-k\\ k
                                \end{array} \right)\mathcal{L}^{p-1-k}\end{eqnarray}
The solution of the corresponding homogenous equation in
accordance with the prescribed initial conditions is then
\begin{eqnarray}\nonumber f_p^{(H)}(t)=\sum\limits_{k=0}^{\left[\frac{p-2}{2}\right]}\alpha^{k+1}\beta^{p-2-2k} \left(\begin{array}{c}
                                p-2-k\\ k
                                \end{array} \right)\mathcal{L}^{p-1-k}\left(f_0(t)\right)+\\\sum\limits_{k=0}^{\left[\frac{p-1}{2}\right]}\alpha^k\beta^{p-1-2k} \left(\begin{array}{c}
                                p-1-k\\ k
                                \end{array} \right)\mathcal{L}^{p-1-k}\left(f_1(t)\right)\end{eqnarray}
Exploiting our central theorem (2), we may claim that
                                \begin{eqnarray}f_p^*=\sum\limits_{m=1}^{p-1}\beta_{p-m}f_{m+1}(0)=
\sum\limits_{m=1}^{p-2}\beta_{p-m}f_{m+1}(0)+f_p(0)\end{eqnarray}
is the particular solution of the nonhomogenous equation for which
$f_0=f_1=0, (\forall) t$. Equivalently, we have that
 \begin{eqnarray}f_p^*= \sum\limits_{m=1}^{p-2}\sum\limits_{k=0}^{\left[\frac{p-m-1}{2}\right]}\alpha^k\beta^{p-m-1-2k} \left(\begin{array}{c}
                                p-m-1-k\\ k
                                \end{array} \right)\mathcal{L}^{p-m-1-k}\left(f_{m+1}(0)\right)+f_p(0) \end{eqnarray}
But, as shown in the Appendix, we may prove that
\begin{eqnarray}\mathcal{L}^n(f(t))=\int\limits_0^tdt_n\int\limits_0^{t_n}dt_{n-1}\int\limits_0^{t_{n-1}}dt_{n-2}\dots\int\limits_0^{t_3}dt_2\int\limits_0^{t_2}f(t_1)dt_1=\frac{1}{(n-1)!}\int\limits_0^t(t-\tau)^{n-1}f(\tau)d\tau,\end{eqnarray}
so that we may write down  that
\begin{eqnarray}f_p^{(H)}(t)=\sum\limits_{k=0}^{\left[\frac{p-2}{2}\right]}\alpha^{k+1}\beta^{p-2-2k} \left(\begin{array}{c}
                                p-2-k\\ k
                                \end{array} \right)\frac{1}{(p-2-k)!}\int\limits_0^t(t-\tau)^{p-2-k}f_0(\tau)d\tau+\nonumber\\\sum\limits_{k=0}^{\left[\frac{p-1}{2}\right]}\alpha^k\beta^{p-1-2k} \left(\begin{array}{c}
                                p-1-k\\ k
                                \end{array} \right)\frac{1}{(p-2-k)!}\int\limits_0^t(t-\tau)^{p-2-r}f_1(\tau)d\tau
                                \end{eqnarray}and

                                \begin{eqnarray}\nonumber f_p^*=\sum\limits_{m=1}^{p-2}\sum\limits_{k=0}^{\left[\frac{p-m-1}{2}\right]}\alpha^k\beta^{p-m-1-2k} \left(\begin{array}{c}
                                p-m-1-k\\ k
                                \end{array} \right)\frac{1}{(p-m-2-k)!}\int\limits_0^t(t-\tau)^{p-m-2-k}f_{m+1}(0)d\tau+f_p(0)\end{eqnarray}
Hence, the general solution of the proposed integral-difference
equation is
\begin{eqnarray}f_p(t)=\sum\limits_{k=0}^{\left[\frac{p-2}{2}\right]}\alpha^{k+1}\beta^{p-2-2k} \left(\begin{array}{c}
                                p-2-k\\ k
                                \end{array} \right)\frac{1}{(p-2-k)!}\int\limits_0^t(t-\tau)^{p-2-k}f_0(\tau)d\tau+\nonumber\\\sum\limits_{k=0}^{\left[\frac{p-1}{2}\right]}\alpha^k\beta^{p-1-2k} \left(\begin{array}{c}
                                p-1-k\\ k
                                \end{array}
                                \right)\frac{1}{(p-2-k)!}\int\limits_0^t(t-\tau)^{p-2-k}f_1(\tau)d\tau+\\ \sum\limits_{m=1}^{p-2}\sum\limits_{k=0}^{\left[\frac{p-m-1}{2}\right]}\alpha^k\beta^{p-m-1-2k} \left(\begin{array}{c}
                                p-m-1-k\\ k
                                \end{array}
                                \right)\frac{t^{p-m-1-k}}{(p-m-2-k)!(p-m-1-k)!}f_{m+1}(0)+f_p(0)
                                \end{eqnarray}

\section{CONCLUSIVE REMARKS}

The novel and  mean theoretical result of this paper is expressed
by theorem (2) with which we  demonstrate that eq. \eqref{SOLP}
provides a particular solution of eq. \eqref{NHE}. This result
together with theorem (1)  completes the resolution of this
equation
 enabling us to write down  formula \eqref{solNH}
for its general solution. The operator character of eq.
\eqref{NHE} and, as a consequence, the presence of generally
noncommuting coefficients is the key to understand why such an
equation may represent the canonical form of equations seemingly
not related each other. The consequences of eq. \eqref{solNH} and
the applications  reported in this paper, besides being
interesting in their own, demonstrate indeed both the wide
applicability of eq. \eqref{NHE} as well as the practical
usefulness of its resolutive formula.

\appendix
\section{APPENDIX}
For the sake of completeness we here report a proof of the
well-known follwing identity
\begin{eqnarray}\label{formula}\int\limits_0^tdt_n\int\limits_0^{t_n}dt_{n-1}\int\limits_0^{t_{n-1}}dt_{n-2}\dots\int\limits_0^{t_3}dt_2\int\limits_0^{t_2}f(t_1)dt_1=\frac{1}{(n-1)!}\int\limits_0^t(t-\tau)^{n-1}f(\tau)d\tau,\end{eqnarray}
where $f(t)$ is a $C^{\infty}$-function. The mathematical
induction procedure will be exploited. \\ For $n=1$ the above
formula becomes the identity
\begin{eqnarray}\int\limits_0^tf(t_1)dt_1=\int\limits_0^tf(\tau)d\tau\end{eqnarray}
Let's suppose that the formula \ref{formula} holds for any $r\leq
n$ and we  prove it validity for $n+1$. To proceed, it is
convenient to define
\begin{eqnarray}F_r(t)=\int\limits_0^tdt_r\int\limits_0^{t_r}dt_{r-1}\int\limits_0^{t_{r-1}}dt_{r-2}\dots\int\limits_0^{t_3}dt_2\int\limits_0^{t_2}f(t_1)dt_1\end{eqnarray}
writing down eq.\eqref{formula} as follows
$$F_n(t)=\frac{1}{(n-1)!}\int\limits_0^t(t-\tau)^{n-1}f(\tau)d\tau$$
It is obviously  that by definition
\begin{eqnarray}F_{r+1}(t)=\int\limits_0^tF_r(t_{r+1})dt_{r+1}, \quad r\geq 1\end{eqnarray}
By the induction hypothesis we have that
\begin{eqnarray}F_{n+1}'=F_n(t)=\frac{1}{(n-1)!}\int\limits_0^t(t-\tau)^{n-1}f(\tau)d\tau\end{eqnarray}
Our problem becomes the resolution of the following Cauchy problem
\begin{eqnarray}\label{A}
\left\{\begin{array}{rl}
      F_{n+1}'(t)=\frac{1}{(n-1)!}\int\limits_0^t(t-\tau)^{n-1}f(\tau)d\tau\\
       F_{n+1}(0)=0\end{array}\right.
\end{eqnarray}
From the well-known Leibnitz identity
\begin{eqnarray}\frac{d}{dx}\int\limits_{a(x)}^{b(x)}f(x,y)dy=\int\limits_{a(x)}^{b(x)}\frac{\partial f(x,y)}{\partial x}dy+b'(x)f(x,b(x))-a'(x)f(x,a(x))\end{eqnarray}
we have that
\begin{eqnarray}\frac{1}{(n-1)!}\int\limits_0^t(t-\tau)^{n-1}f(\tau)d\tau=\frac{1}{n(n-1)!}\int\limits_0^t\frac{\partial \left[(t-\tau)^{n}f(\tau)\right]}{\partial t}d\tau =\frac{d}{dt}\left[\frac{1}{n!}\int\limits_0^t(t-\tau)^nf(\tau)d\tau\right]\end{eqnarray}
Hence, the differential equation for $F_{n+1}(t)$ may be rewritten
as
\begin{eqnarray}F_{n+1}'(t)=\frac{d}{dt}\left[\frac{1}{n!}\int\limits_0^t(t-\tau)^nf(\tau)d\tau\right]\end{eqnarray}
such that the Cauchy problem \eqref{A} has the solution
\begin{eqnarray}F_{n+1}(t)=\frac{1}{n!}\int\limits_0^t(t-\tau)^nf(\tau)d\tau\end{eqnarray}

\end{document}